\documentclass[11pt]{article}
\usepackage{CJK,amsmath, amssymb, amsthm,color,mathrsfs}
\usepackage{setspace}
\usepackage{graphicx}
\usepackage{verbatim}
\usepackage{float}

\newtheorem{theorem}{Theorem}[section]

\newtheorem{example}{Example}[section]
\newtheorem{remark}{Remark}[section]

\begin{document}

\onehalfspacing

\title{Nonlinear Regression without i.i.d. Assumption}

\author{Qing Xu\footnote{UniDT. Email: qing.xu@unidt.com} and Xiaohua (Michael) Xuan\footnote{UniDT. Email: michael.xuan@unidt.com}~\footnote{This work is partially supported by Smale Institute.}}

\date{}

\maketitle

\begin{abstract}
In this paper, we consider a class of nonlinear regression problems without the assumption of being independent and identically distributed. We propose a correspondent mini-max problem for nonlinear regression and give a numerical algorithm. Such an algorithm can be applied in regression and machine learning problems, and yield better results than traditional least square and machine learning methods.
\end{abstract}
\bigskip

\section{Introduction}

In statistics, linear regression is a linear approach for modelling the relationship between an explaining variable $y$ and one or more explanatory variables denoted by $x$.
\begin{equation}\label{linearmodel}
y=w^Tx+b+\varepsilon.
\end{equation}

The parameters $w,b$ can be estimated via the method of least square by the following famous theorem.

\begin{theorem}
Suppose $\{(x_i,y_i)\}_{i=1}^m$ are drawn from the linear model (\ref{linearmodel}) with error terms $\varepsilon_1,\varepsilon_2,\cdots,\varepsilon_m$ being independent Gaussian variables with mean $0$ and variance $\sigma^2$. Then the result of least square is
$$(w_1,w_2,\cdots,w_d,b)^T =A^+c.$$
Here,
\[
A =
\begin{pmatrix}
x_{11}&x_{12}&\cdots&x_{1d}&1\\
x_{21}&x_{22}&\cdots&x_{2d}&1\\
\cdots&\cdots&\cdots&\cdots&\cdots\\
x_{m1}&x_{m2}&\cdots&x_{md}&1\\
\end{pmatrix},\quad
c=
\begin{pmatrix}
y_1\\
y_2\\
\cdots\\
y_m\\
\end{pmatrix}.
\]
$A^+$ is the Moore-Penrose inverse\footnote{For the definition and property of Moore-Penrose inverse, see \cite{ben}.} of $A$.
\end{theorem}

In the above theorem, the errors $\varepsilon_1,\varepsilon_2,\cdots,\varepsilon_m$ are assumed to be independent Gaussian variables. Therefore, $y_1,y_2,\cdots,y_m$ are also independent Gaussian variables.

When the i.i.d. (independent and identically distributed) assumption is not satisfied, the usual method of least square does not work well. This can be illustrated by the following example.

\begin{example}\label{linear}
Denote by $\mathcal{N}(\mu,\sigma^2)$ the normal distribution with mean $\mu$ and variance $\sigma^2$ and denote by $\delta_{c}$ the Dirac distribution, i.e.,
\[\delta_c(A)=\left\{
\begin{split}
&1\quad c\in A,\\
&0\quad c\notin A.
\end{split}
\right.
\]

Suppose the sample data are generated by
$$y_i=1.75*x_i+1.25+\varepsilon_i,\quad i=1,2,\cdots,1517,$$
where
$$\varepsilon_1,\cdots,\varepsilon_{500}\sim\delta_{0.0325},
\quad \varepsilon_{501},\cdots,\varepsilon_{1000}\sim\delta_{0.5525},$$
$$\varepsilon_{1001},\cdots,\varepsilon_{1500}\sim\delta_{-0.27},\quad \varepsilon_{1501},\cdots,\varepsilon_{1517}\sim\mathcal{N}(0,0.2).$$

Thus, the whole sample data are chosen as 
$$(x_{1},y_{1})=(x_{2},y_{2})=\cdots=(x_{500},y_{500})=(0.15,1.48),$$
$$(x_{501},y_{501})=(x_{502},y_{502})=\cdots=(x_{1000},y_{1000})=(0.43,1.45),$$
$$(x_{1001},y_{1001})=(x_{1002},y_{1002})=\cdots=(x_{1500},y_{1500})=(0.04,1.59),$$
$$(x_{1501},y_{1501})=(1.23,3.01),\quad(x_{1502},y_{1502})=(0.63,2.89),\quad(x_{1503},y_{1503})=(1.64,4.54),$$
$$(x_{1504},y_{1504})=(0.98,3.32),\quad(x_{1505},y_{1505})=(1.92,5.0),\quad(x_{1506},y_{1506})=(1.26,3.96),$$
$$(x_{1507},y_{1507})=(1.77,3.92),\quad(x_{1508},y_{1508})=(1.1,2.8),\quad(x_{1509},y_{1509})=(1.22,2.84),$$
$$(x_{1510},y_{1510})=(1.48,4.52),\quad(x_{1511},y_{1511})=(0.71,3.17),\quad(x_{1512},y_{1512})=(0.77,2.59),$$
$$(x_{1513},y_{1513})=(1.89,5.1),\quad(x_{1514},y_{1514})=(1.31,3.17),\quad(x_{1515},y_{1515})=(1.31,2.91),$$
$$\quad(x_{1516},y_{1516})=(1.63,4.02),\quad(x_{1517},y_{1517})=(0.56,1.79).$$

The result of the usual least square is 
$$y = 0.4711*x+1.4258,$$
which is dispalyed in Figure 1.

\begin{figure}[H]
	\center{\includegraphics[width=7cm]{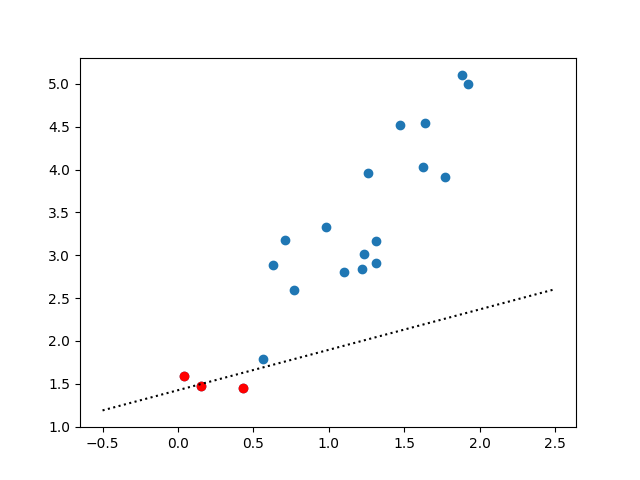}}
	\caption{\small Result of Least Square}
\end{figure}

\end{example}

We can see from Figure 1 that most of the sample data deviates from the regression line. The main reason is that $(x_1,y_1),(x_2,y_2),\cdots,(x_{500},y_{500})$ are the same sample and the i.i.d. condition is violated.

For overcoming the above difficulty, Lin\cite{lin} study the linear regression without i.i.d. condition by using the nonlinear expectation framework laid down by Peng\cite{peng}. They split the training set into several groups and in each group the i.i.d. condition can be satisfied. The average loss is used for each group and the maximum of average loss among groups is used as the final loss function. They show that the linear regression problem under the nonlinear expectation framework is reduced to the following mini-max problem.
\begin{equation}\label{lin}
\min_{w,b}\max_{1\leq j\leq N}\frac{1}{M}\sum_{l=1}^M(w^Tx_{jl}+b-y_{jl})^2.
\end{equation}
They suggest a genetic algorithm to solve this problem. However, such a genetic algorithm does not work well generally.

Motivated by the work of Lin\cite{lin} and Peng\cite{peng}, we consider nonlinear regression problems without the assumption of i.i.d. in this paper. We propose a correspondent mini-max problems and give a numerical algorithm for solving this problem. Meanwhile, problem (\ref{lin}) in Lin's paper can also be well solved by such an algorithm. We also have done some experiments in least square and machine learning problems.

\section{Nonlinear Regression without i.i.d. Assumption}

Nonlinear regression is a form of regression analysis in which observational data are modeled by a nonlinear function which depends on one or more explanatory variables.~(see e.g. \cite{nr})

Suppose the sample data (training set) is
$$S=\{(x_1,y_1),(x_2,y_2),\cdots,(x_m,y_m)\},$$
where $x_i\in X$~and~$y_i\in Y$. $X$ is called the input space and $Y$ is called the output (label) space. The goal of nonlinear regression is to find (learn) a function $g^\theta:X\rightarrow Y$ from the hypothesis space~$\{g^\lambda:X\rightarrow Y|\lambda\in \Lambda\}$~such that~$g^\theta(x_i)$~is as close to $y_i$ as possible.

The closeness is usually characterized by a loss function $\varphi$ such that $\varphi(g^\theta(x_1),y_1,\cdots,g^\theta(x_m),y_m)$ attains its minimum 
if and only if
$$g^\theta(x_i)-y_i=0,\quad 1\leq i\leq m.$$

Then the nonlinear regression problem (learning problem) is reduced to an optimization problem of minimizing $\varphi$.

Following are two kinds of loss functions, namely, the average loss and the maximal loss.

$$\varphi_2=\frac{1}{m}\sum_{j=1}^m(g^\theta(x_j)-y_j)^2.$$
$$\varphi_\infty=\max_{1\leq j\leq m}(g^\theta(x_j)-y_j)^2.$$

The average loss is popular, particularly in machine learning, since it can be conveniently minimized using online algorithms, which process few instances at each iteration. The idea behinds the average loss is to learn a function that performs equally well for each training point. However, when the i.i.d. assumption is not satisfied, the average loss way may become a problem.

To overcome this difficulty, we use the max-mean as the loss function. First, we split the training set into several groups and in each group the i.i.d. condition can be satisfied. Then the average loss is used for each group and the maximum of average loss among groups is used as the final loss function. We propose the following mini-max problem for nonlinear regression problem.
\begin{equation}\label{our}
\min_{\theta}\max_{1\leq j\leq N}\frac{1}{n_j}\sum_{l=1}^{n_j}(g^\theta(x_{jl})-y_{jl})^2.
\end{equation}
Here, $n_j$ is the number of samples in group $j$.

Problem (\ref{our}) is a generalization of problem (\ref{lin}). Next, we will give a numerical algorithm which solves problem (\ref{our}).

\begin{remark}
 Peng and Jin\cite{han} put forward a max-mean method to give the parameter estimation when the usual i.i.d. condition is not satisfied. They show that if $Z_1,Z_2,\cdots,Z_k$ are drawn from the maximal distribution $M_{[\underline{\mu},\overline{\mu}]}$ and are nonlinearly independent, then the optimal unbiased estimation for $\overline{\mu}$ is
$$\max\{Z_1,Z_2,\cdots,Z_k\}.$$
This fact, combined with the Law of Large Numbers (Theorem 19 in \cite{han}) leads to the max-mean estimation of $\mu$. We borrow this idea and use the max-mean as the loss function for the nonlinear regression problem.
\end{remark}

\section{Algorithm}
Problem (\ref{our}) is a mini-max problem. The mini-max problems arise in different kinds of mathematical fields, such as game theory and the worst-case optimization. The general mini-max problem is described as
\begin{equation}\label{minimax}
\min_{u\in\mathbb{R}^n}\max_{v\in V}h(u,v).
\end{equation}
Here,~$h$~is continuous on $\mathbb{R}^n\times V$ and differentiable with respect to $u$.

Problem (\ref{minimax}) was considered theoretically by Klessig and Polak\cite{klessig} in 1973 and Panin\cite{panin} in 1981. Later in 1987, Kiwiel\cite{kiwiel} gave a concrete algorithm for problem (\ref{minimax}). Kiwiel's algorithm deals with the general case in which $V$ is a compact subset of $\mathbb{R}^d$ and the convergence could be slow when the number of parameters is large.

In our case, $V=\{1,2,\cdots,N\}$ is a finite set and we will give a simplified and faster algorithm.

Denote
$$f_j(u)=h(u,j)=\frac{1}{n_j}\sum_{l=1}^{n_j}(g^u(x_{jl})-y_{jl})^2,\quad \Phi(u)=\max_{1\leq j\leq N}f_j(u).$$

Suppose each $f_j$ is differentiable. Now we outline the iterative algorithm for the following discrete mini-max problem
$$\min_{u\in \mathbb{R}^n}\max_{1\leq j\leq N}f_j(u).$$

The main difficulty is to find the descent direction at each iteration point $u_k(k=0,1,\cdots)$ since $\Phi$ is nonsmooth in general.
In light of this, we linearize $f_j$ at $u_k$ and obtain the convex approximation of $\Phi$ as
$$\hat{\Phi}(u)=\max_{1\leq j\leq N}\{f_j(u_k)+\langle \nabla f_j(u_k),u-u_k\rangle\}.$$

Next step is to find $u_{k+1}$, which minimizes $\hat{\Phi}(u)$. In general, $\hat{\Phi}$ is not strictly convex with respect to $u$, thus may not admit a minimum. So a regularization term is added and the minimization problem becomes
$$\min_{u\in\mathbb{R}^n}\left\{\hat{\Phi}(u)+\frac{1}{2}\|u-u_k\|^2\right\}.$$

By setting $d=u-u_k$, the above can be converted to the following form
\begin{equation}\label{kk}
\min_{d\in\mathbb{R}^n}\left\{\max_{1\leq j\leq N}\{f_j(u_k)+\langle \nabla f_j(u_k),d\rangle\}+\frac{1}{2}\|d\|^2\right\},
\end{equation}
which is equivalent to
\begin{equation}\label{prim}
\min_{d,a}\quad \left(\frac{1}{2}\|d\|^2+a\right)
\end{equation}
\begin{equation}\label{rest}
\text{s.t.}~f_j(u_k)+\langle \nabla f_j(u_k),d\rangle\leq a,~\forall~1\leq j\leq N.
\end{equation}

Problem (\ref{prim})-(\ref{rest}) is a semi-definite QP (quadratic programming) problem. When $n$ is large, the popular QP algorithms (such as active-set method) are time-consuming. So we turn to the dual problem.

\begin{theorem}\label{th1}
Denote $G=\nabla f\in \mathbb{R}^{N\times n},f=(f_1,\cdots,f_N)^T$. If $\lambda$ is the solution of the following QP problem
\begin{equation}\label{dual1}
\min_\lambda \left(\frac{1}{2}\lambda^TGG^T\lambda -f^T\lambda \right)
\end{equation}
\begin{equation}\label{dual2}
\mathrm{s.t.}~\sum_{i=1}^{N}\lambda_i =1,\lambda_i\geq0.
\end{equation}
Then $d=-G^T\lambda$ is the solution of problem (\ref{prim})-(\ref{rest}).
\end{theorem}

\begin{proof}
See Appendix.
\end{proof}

\begin{remark}
Problem (\ref{dual1})-(\ref{dual2})	can be solved by many standard methods, such as active-set method(see e.g.\cite{no}). The dimension of the dual problem(\ref{dual1})-(\ref{dual2}) is $N$(number of groups), which is independent of $n$(number of parameters). Hence, the algorithm is fast and stable, especially in deep neural networks.
\end{remark}

Set $d_k=-G^T\lambda$. Next Theorem shows that $d_k$ is a descent direction.

\begin{theorem}\label{th2}
If $d_k\neq 0$, then there exists $t_0>0$ such that
	$$\Phi(u_{k}+td_k)<\Phi(u_k),\quad \forall~ t\in (0,t_0).$$ 
\end{theorem}

\begin{proof}
	See Appendix.
\end{proof}

For a function $F$, the directional derivative of $F$ at $x$ in a direction $d$ is defined as 
$$F'(x;d):=\lim_{t\rightarrow 0+}\frac{F(x+td)-F(x)}{t}.$$

The necessary optimality condition for a function $F$ to attain its minimum (see \cite{dm}) is 
$$F'(x;d)\geq 0,~\forall d\in\mathbb{R}^n.$$ 
$x$ is called a stationary point of $F$.

Theorem \ref{th2} shows that when $d_k\neq 0$, we can always find a descent direction. Next Theorem reveals that when $d_k=0$, $u_k$ is a stationary point. 

\begin{theorem}\label{th3}
If $d_k=0$, then $u_k$ is a stationary point of $\Phi$, i.e., 
$$\Phi'(u_k;d)\geq 0,~\forall d\in\mathbb{R}^n.$$ 
\end{theorem}
\begin{proof}
See Appendix.
\end{proof}
\begin{remark}
When each $f_j$ is a convex function, $\Phi$ is also a convex function. Then the stationary point of $\Phi$ becomes the global minimum point.
\end{remark}

With $d_k$ being the descent direction, we can use line search to find the appropriate step size and update the iteration point.
\bigskip

Now let us conclude the above discussion by giving the concrete steps of the algorithm for the following mini-max problem.
\begin{equation}\label{e1}
\min_{u\in \mathbb{R}^n}\max_{1\leq j\leq N}f_j(u).
\end{equation}
\bigskip

\noindent\textbf{Algorithm.}

\noindent\textbf{Step 1. Initialization}

Select arbitrary~$u_0\in \mathbb{R}^n$.~Set~$k=0$, termination accuracy $\xi = 10^{-8}$, gap tolerance $\delta=10^{-7}$ and step size factor $\sigma=0.5$.
\bigskip

\noindent\textbf{Step 2. Finding Descent Direction }

Assume that we have chosen~$u_k$. Compute the Jacobian matrix

$$G=\nabla f(u_k)\in\mathbb{R}^{N\times n},$$
where
$$f(u)=(f_1(u),\cdots,f_N(u))^T.$$

Solve the following quadratic programming problem with gap tolerance $\delta$.(see e.g.\cite{no})

$$\min_{\lambda}\left(\frac{1}{2}\lambda^TGG^T\lambda -f^T\lambda \right)$$
$$\mathrm{s.t.}~\sum_{i=1}^{N}\lambda_i =1,\lambda_i\geq0.$$

Take $d_k=-G^T\lambda$. If $\|d_k\|<\xi$, stop. Otherwise, goto Step 3.
\bigskip

\noindent\textbf{Step 3. Line Search}

Find the smallest natural number $j$ such that
$$\Phi(u_k+\sigma^j d_k)<\Phi(u_k).$$

Take $\alpha_k = \sigma^j$ and set $u_{k+1} = u_k +\alpha_kd_k,~k=k+1$. Goto Step 2.
\bigskip

\section{Experiments}

\subsection{The Linear Regression Case}

Example \ref{linear} can be numerically well solved by the above algorithm with
$$f_j(w,b)=(wx_j+b-y_j)^2,\quad j=1,2,\cdots,1517.$$
The corresponding optimization problem is
$$\min_{w,b}\max_{1\leq j\leq 1517}(wx_j+b-y_j)^2.$$
The numerical result using the algorithm in section 3 is
$$y = 1.7589*x+1.2591.$$

\begin{figure}[H]
	\center{\includegraphics[width=7cm]{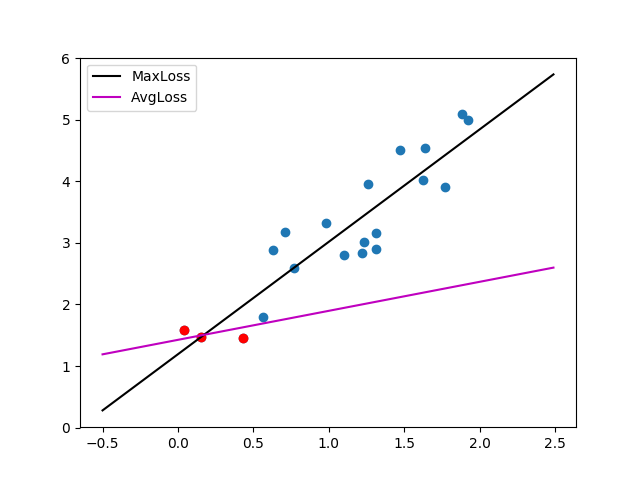}}
	\caption{\small Results of two methods}
\end{figure}

Figure 2 summarize the result. It can be seen that the mini-max method (black line) performs better than the traditional least square method (pink line).

Next, we compare the two methods and both MSE (mean squared error) and MAE (mean absolute error) are used to measure of the estimation of $w$ and $b$.

$$\text{MSE}:=\frac{1}{2}\left((w-\hat{w})^2+(b-\hat{b})^2\right).$$
$$\text{MAE}:=\frac{1}{2}\left(|w-\hat{w}|+|b-\hat{b}|\right).$$

\begin{table}[H]\centering
	\begin{tabular}{lll}
		\hline\hline
		Method             & MSE      & MAE      \\
		\hline
		Traditional Method   & 0.8178       & 0.6439         \\
		Mini-max Method      & 0.0154     & 0.0924           \\
		\hline
	\end{tabular}
	\caption{Comparisons of the two methods}\label{tab:hh}
\end{table}

We can see from the above table that mini-max method outperform the traditional method in both MSE and MAE.

Lin et al.\cite{lin} have mentioned that the above problem can be solved by genetic algorithm. However, the genetic algorithm is heuristic and unstable especially when the number of group is large. In contrast, our algorithm is fast and stable and the convergence is proved.

\subsection{The Machine Learning Case}

We further test the proposed method by using CelebFaces Attributes Dataset (CelebA)\footnote{see http://mmlab.ie.cuhk.edu.hk/projects/CelebA.html} and implement the mini-max algorithm with deep learning approach. The dataset CelebA has 202599 face images among which 13193(6.5\%) has eyeglass. The objective is eyeglass detection. We use a single hidden layer neural network to compare the two different methods.

We randomly choose 20000 pictures as the training set among which 5\% has eyeglass labels. For the traditional method, the 20000 pictures are used as a whole. For the mini-max method, we separate the 20000 pictures into 20 groups. Only 1 group contains eyeglass pictures while the other 19 groups do not contain eyeglass pictures. In this way, the whole mini-batch is not i.i.d. while each subgroup is expected to be i.i.d..

The tradition method uses the following loss

$$\text{loss}=\frac{1}{20000}\sum_{i=1}^{20}\sum_{j=1}^{1000}(\sigma(Wx_{ij}+b)-y_{ij})^2.$$

The mini-max method uses the maximal group loss

$$\text{loss}=\max_{1\leq i\leq 20}\frac{1}{1000}\sum_{j=1}^{1000}(\sigma(Wx_{ij}+b)-y_{ij})^2.$$

Here, $\sigma$ is an activation function in deep learning such as the sigmoid function
$$\sigma(x)=\frac{1}{1+e^{-x}}.$$

\begin{figure}[H]
	\center{\includegraphics[width=7cm]{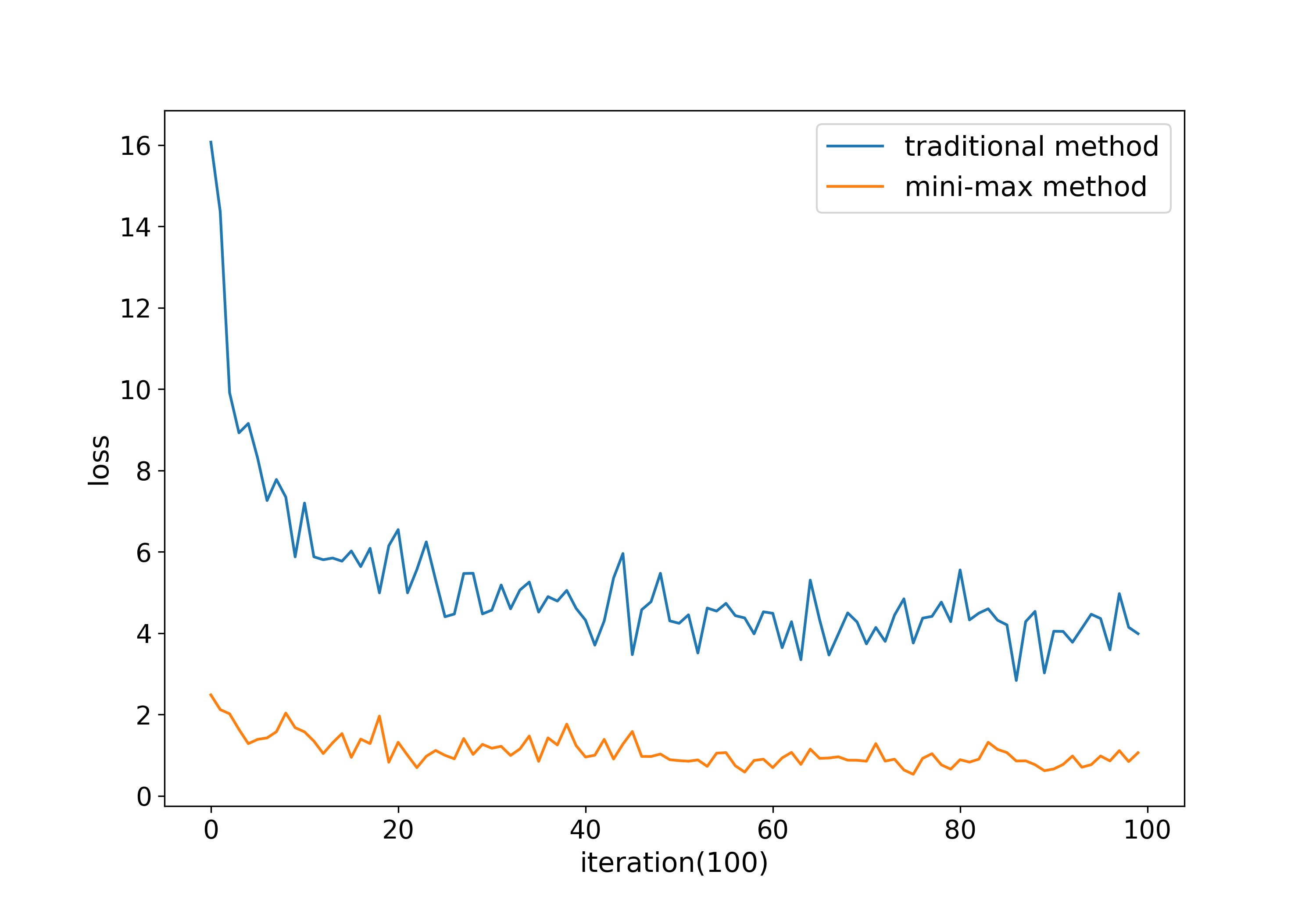}}
	\caption{\small Loss curves of the two methods.}
\end{figure}

We perform the two methods for 100 iterations. We can see from Figure 3 that the mini-max method converges much faster than the traditional method. Figure 4 also shows that the mini-max method performs better than the traditional method in accuracy. (Suppose the total number of the test set is $n$, and $m$ of them are classified correctly. Then the accuracy is defined to be $m/n$.)

\begin{figure}[H]
	\center{\includegraphics[width=7cm]{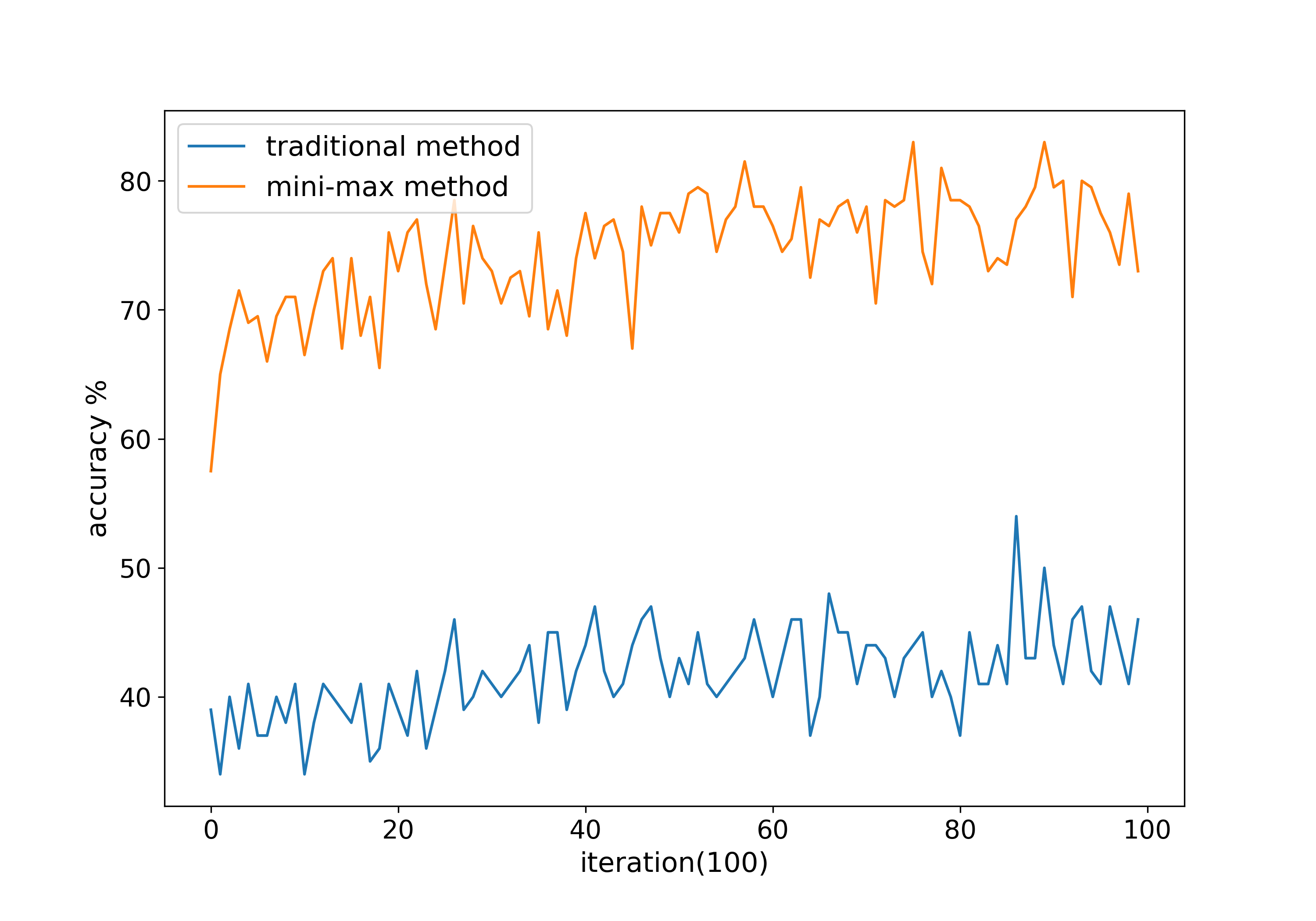}}
	\caption{\small Accuracy curves of the two methods.}
\end{figure}

The average accuracy for the mini-max method is 74.52\% while the traditional method is 41.78\%. Thus, in the deep learning approach with single layer, mini-max method helps to speed up convergence on unbalanced training data and improves accuracy as well. We also expect improvement with multi-layer deep learning approach. 

\section{Conclusion}

In this paper, we consider a class of nonlinear regression problems without the assumption of being independent and identically distributed. We propose a correspondent mini-max problem for nonlinear regression and give a numerical algorithm. Such an algorithm can be applied in regression and machine learning problems, and yield better results than least square and machine learning methods.

\section*{Acknowledgement}
The authors would like to thank Professor Shige Peng for useful discussions. We especially thank Xuli Shen for performing the experiment in machine learning case.

\section*{Appendix}

\subsection*{1. Proof of Theorem \ref{th1}}

	Consider the Lagrange function
	$$L(d,a;\lambda)=\frac{1}{2}\|d\|^2+a+\sum_{j=1}^{N}\lambda_j(f_j(u_k)+\langle \nabla f_j(u_k),d\rangle - a).$$	
	It is easy to verify that problem (\ref{prim})-(\ref{rest}) is equivalent to the following minimax problem.
	$$\min_{d,a}\max_{\lambda \geq 0}L(d,a;\lambda).$$	
	By strong duality theorem(see e.g. \cite{cvx}),
	$$\min_{d,a}\max_{\lambda \geq 0}L(d,a;\lambda)=\max_{\lambda \geq 0}\min_{d,a}L(d,a;\lambda).$$	
	Set $e=(1,1,\cdots,1)^T$, the above problem is equivalent to
	$$\max_{\lambda \geq 0}\min_{d,a}\left(\frac{1}{2}\|d\|^2+a+\lambda^T(f+Gd-ae)\right).$$	
	Note that
	\begin{equation*}
		\frac{1}{2}\|d\|^2+a+\lambda^T(f+Gd-ae)=\frac{1}{2}\|d\|^2+\lambda^T(f+Gd)+a(1-\lambda^Te).
	\end{equation*}	
	If $1-\lambda^Te\neq 0$, then the above is $-\infty$. Thus, we must have $1-\lambda^Te = 0$ when the maximum is attained. The problem is converted to
	$$\max_{\lambda_i \geq 0,\sum_{i=1}^N\lambda_i=1}\min_{d}\left(\frac{1}{2}\|d\|^2+\lambda^TGd+\lambda^Tf\right).$$
	The inner minimization problem has solution $d=-G^T\lambda$ and the above problem is reduced to
	$$\min_\lambda \left(\frac{1}{2}\lambda^TGG^T\lambda -f^T\lambda \right)$$
	$$\mathrm{s.t.}~\sum_{i=1}^{N}\lambda_i =1,\lambda_i\geq0.$$

\subsection*{2. Proof of Theorem \ref{th2}}
	
	Denote $u=u_k,d=d_k$. For $0<t<1$,
	\begin{align*}
	&\Phi(u+td)-\Phi(u)\\
	=&\max_{1\leq j\leq N}\{f_j(u+td)-\Phi(u)\}\\
	=&\max_{1\leq j\leq N}\{f_j(u)+t\langle \nabla f_j(u),d\rangle -\Phi(u)+o(t)\}\\
	\leq&\max_{1\leq j\leq N}\{f_j(u)+t\langle \nabla f_j(u),d\rangle -\Phi(u)\}+o(t)\\
	=&\max_{1\leq j\leq N}\{t(f_j(u)+\langle \nabla f_j(u),d\rangle -\Phi(u))+(1-t)(f_j(u)-\Phi(u))\}+o(t)\\
	&\quad\quad\quad \left(\text{Note that } f_j(u)\leq \Phi(u)=\max_{1\leq k\leq N}f_k(u)\right)\\
	\leq &t\max_{1\leq j\leq N}\{f_j(u)+\langle \nabla f_j(u),d\rangle -\Phi(u)\}+o(t).
	\end{align*}	
	Since $d$ is the solution of problem (\ref{kk}), we have that
	\begin{align*}
	&\max_{1\leq j\leq N}\left\{f_j(u)+\langle \nabla f_j(u),d\rangle+\frac{1}{2}\|d\|^2 \right\}\\
	\leq &\max_{1\leq j\leq N}\left\{f_j(u)+\langle \nabla f_j(u),0\rangle+\frac{1}{2}\|0\|^2 \right\}\\
	=& \max_{1\leq j\leq N}\{f_j(u)\}\\
	=&\Phi(u).
	\end{align*}	
	Therefore,
	\begin{align*}
	&\max_{1\leq j\leq N}\{f_j(u)+\langle \nabla f_j(u),d\rangle -\Phi(u)\}\leq -\frac{1}{2}\|d\|^2.\\
    \Rightarrow~ &\Phi(u+td)-\Phi(u)\leq -\frac{1}{2}t\|d\|^2+o(t).\\
	\Rightarrow~&\frac{\Phi(u+td)-\Phi(u)}{t}\leq -\frac{1}{2}\|d\|^2+o(1).\\
	\Rightarrow~&\limsup_{t\rightarrow0+}\frac{\Phi(u+td)-\Phi(u)}{t}\leq -\frac{1}{2}\|d\|^2<0.
	\end{align*}
	For $t>0$ small enough, we have that
	$$\Phi(u+td)<\Phi(u).$$	

\subsection*{3. Proof of Theorem \ref{th3}}
Denote $u=u_k$. Then, $d_k=0$ means that $\forall d$,
\begin{equation}\label{abb}
\max_{1\leq j\leq N}\{f_j(u)+\langle \nabla f_j(u),d\rangle\}+\frac{1}{2}\|d\|^2\geq \max_{1\leq j\leq N}f_j(u).
\end{equation}
Denote
$$M=\max_{1\leq j\leq N}f_j(u).$$
Define
$$\Theta=\Big\{j|f_j(u) = M,j=1,2,\cdots,N \Big\}.$$
Then (see \cite{dm})
\begin{equation}\label{kkk}
\Phi'(u;d)= \max_{j\in\Theta}\langle \nabla f_j(u),d\rangle. 
\end{equation}
When $\|d\|$ is small enough, we have that
\begin{align*}
&\max_{1\leq j\leq N}\{f_j(u)+\langle \nabla f_j(u),d\rangle\}\\
=&\max_{j\in\Theta}\{f_j(u)+\langle \nabla f_j(u),d\rangle\}\\
=&M+\max_{j\in\Theta}\langle \nabla f_j(u),d\rangle.
\end{align*}
In view of (\ref{abb}), we have that for $\|d\|$ small enough,
$$\max_{j\in\Theta}\langle \nabla f_j(u),d\rangle+\frac{1}{2}\|d\|^2\geq0.$$
For any $d_1\in\mathbb{R}^n$, by taking $d=rd_1$ with sufficient small $r>0$, we have that
$$\max_{j\in\Theta}\langle \nabla f_j(u),rd_1\rangle+\frac{r^2}{2}\|d_1\|^2\geq0.$$
$$\max_{j\in\Theta}\langle \nabla f_j(u),d_1\rangle+\frac{r}{2}\|d_1\|^2\geq0.$$
Let $r\rightarrow0+$,
$$\max_{j\in\Theta}\langle \nabla f_j(u),d_1\rangle\geq0.$$
Thus we fulfill the proof by combining with the fact (\ref{kkk}).

\bibliographystyle{unsrt}

\end{document}